\documentclass[runningheads]{llncs}
\pdfoutput=1
\usepackage{amsmath}
\usepackage {amsfonts}
\usepackage{graphicx}
\usepackage{verbatim}
\usepackage{url}
\usepackage{times}
\usepackage[table,xcdraw]{xcolor}
\usepackage{setspace}
\usepackage[usenames,dvipsnames]{pstricks}
\usepackage{epsfig}
\usepackage{pst-grad} 
\usepackage{pst-plot} 
\usepackage[T1]{fontenc} 
\usepackage{soul}
\usepackage{wrapfig}
\usepackage{cutwin}

\newcommand{\cL} {$\mathsf{L}$}

\usepackage{enumitem}
\setlist[description]{font=\normalfont\itshape\space}

\usepackage{xargs}  
\usepackage{todonotes}
\newcommandx{\info}[2][1=]{\todo[linecolor=blue,backgroundcolor=blue!25,bordercolor=blue,#1]{#2}}

\newcommand{\cReverseL} {\text{\reflectbox{$\mathsf{L}$}}}
\newcommand{\cReverseGamma}{\raisebox{\depth}{\rotatebox{180}{$\mathsf{L}$}}}
\newcommand{\cGamma}{\reflectbox{\raisebox{\depth}{\rotatebox{180}{$\mathsf{L}$}}}}

\newcommand{\cMLE}{monotone $\mathsf{L}$-embedding }
\newcommand{\cMLG}{monotone $\mathsf{L}$-graph }
\newcommand{\cLE}{$\mathsf{L}$-embedding }
\newcommand{\cLG}{$\mathsf{L}$-graph }
\newcommand{\LLL}{\mathbb{L}}

\author{Abu Reyan Ahmed, Felice De Luca, Sabin Devkota, Alon Efrat,  Md Iqbal Hossain, Stephen Kobourov, Jixian Li,   Sammi Abida Salma, \and Eric Welch}

\authorrunning{Ahmed et al.}

\institute {Department of Computer Science\\ University of Arizona}

\titlerunning{ $\mathsf{L}$-Graphs and Monotone $\mathsf{L}$-Graphs}

\begin{document}
	
	\title{ $\mathsf{L}$-Graphs and Monotone $\mathsf{L}$-Graphs}

	\pagenumbering{arabic}
	\maketitle 
	
	\begin{abstract}
		An $\mathsf{L}$-segment consists of a horizontal and a vertical straight line which form an $\mathsf{L}$. 
In an $\mathsf{L}$-embedding of a graph, each vertex is represented by an $\mathsf{L}$-segment, and two segments intersect each other if and only if the corresponding vertices are adjacent in the graph. If the corner of each $\mathsf{L}$-segment in an $\mathsf{L}$-embedding lies on a straight line, we call it a monotone $\mathsf{L}$-embedding.  In this paper we give a full characterization of monotone $\mathsf{L}$-embeddings by  introducing a new class of graphs which we call ``non-jumping" graphs. We show that a graph admits a monotone $\mathsf{L}$-embedding if and only if the graph is a non-jumping graph. Further, we show that outerplanar graphs, convex bipartite graphs, interval graphs, 3-leaf power graphs, and complete graphs are subclasses of non-jumping graphs. Finally, we show that distance-hereditary  graphs and $k$-leaf power graphs ($k\le 4$) admit $\mathsf{L}$-embeddings.
		
	\end{abstract}

    \section{Introduction}\label{se:introduction}
\label{Introduction}
	
Geometric representations of graphs have been used to reveal intriguing connections between the continuous world of geometry and the discrete world of combinatorial structures. Having a geometric representation is much more than just a way to display a graph, as it reveals underlying structures that can often be described only using geometry.  
A good geometric representation of a graph also leads to algorithmic solutions for purely graph-theoretic questions that, on the surface, do not seem to have anything to do with geometry. Examples of this include rubber band representations in planarity testing~\cite{linial1988}, circle-contact representations in balanced graph partitioning and approximating optimal bisection~\cite{spielman1996}, volume-respecting embeddings in approximation algorithms for graph bandwidth~\cite{Feige1998},  
and orthogonal representations in algorithms for graph connectivity and graph coloring~\cite{lovasz1999}.

In an intersection representation of a graph, vertices are geometric objects (e.g., curves) and edges are realized by intersections (e.g., curve crossings).  
Among the most general types of intersection graphs are \textit{string-graphs},
or graphs that admit a \textit{string representation}, 
in which vertices are represented by arbitrary curves in the plane; see Fig.~\ref{fig:i_graph_ex}(a-b). String-graphs find a practical application in the modeling of integrated thin film RC circuits, where some pairs of conductors in a circuit can cross~\cite{sinden1966topology}.
The class of \textit{$k$-string-graphs} contains the graphs that have a string representation with at most $k$ intersections between two strings, where $k\ge0$.
 Not every graph is a string-graph; for instance, the full subdivision graph of the graph $K_5$ does not have a string representation; see Fig.~\ref{fig:i_graph_ex}(f). 
 
\begin{figure}[tbh]
	\hfil \includegraphics[width=1.05\textwidth]{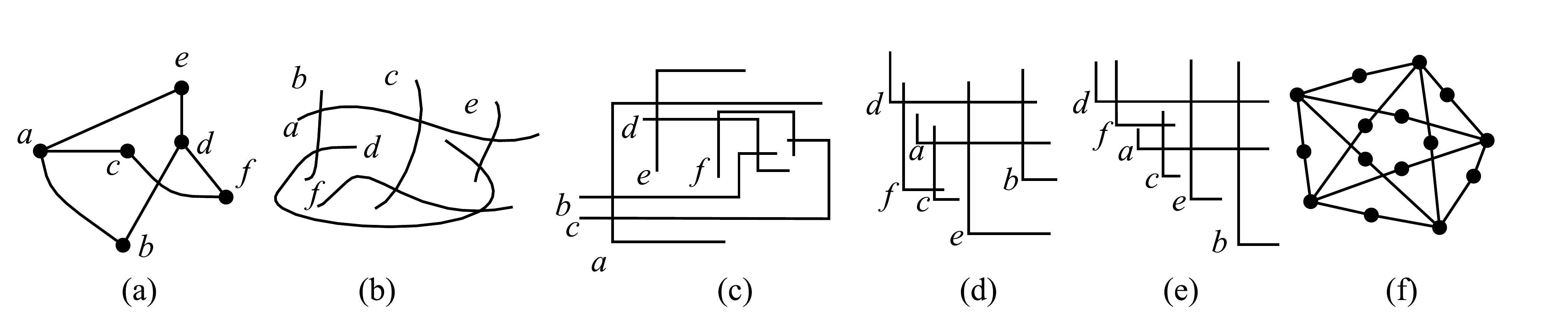} \hfil
	\vspace{-.5cm}\caption{(a) A graph $G$, (b) a string representation of $G$ , and (c) a B$_{2}$-VPG representation of $G$. (d) an \cLE of $G$, (e) a \cMLE of $G$, and (f) a graph that is not an \cLG.}	\label{fig:i_graph_ex}
\end{figure}

Planar graphs are known to be $1$-string graphs~\cite{chalopin2007planar,chalopin2010planar,ehrlich1976intersection}. 
Chalopin and Gon\c{c}alves strengthen this result by proving a conjecture of Scheinerman~\cite{scheinerman1984intersection} that every planar graph has a \textit{segment representation}~\cite{chalopin2009every}, where the segments have arbitrary slopes and intersect at arbitrary angles. The class of segment (\textit{SEG}) graphs is included in the class of $1$-string-graphs.
The recognition of string-graphs is NP-hard~\cite{kratochvil1991string,middendorf1993weakly}.

Another widely-studied class of graphs is the \textit{Vertex Path Grid} (\textit{VPG}) class, introduced by Asinowski \textit{et al.}~\cite{asinowski2011string,asinowski2012vertex}. The class of \textit{$k$-Bends VPG} (\textit{B$_k$-VPG}) graphs restricts the number of bends of the orthogonal paths to $k$, with $k \ge 0$; see Fig.~\ref{fig:i_graph_ex}(c). The class of B$_k$-VPG-graphs is equivalent to the class of string-graphs~\cite{asinowski2011string,asinowski2012vertex}. Chaplick \textit{et al.}~\cite{chaplick2012bendbounded} showed that for every fixed $k$, the recognition of B$_k$-VPG-graph is NP-complete even when the input graph is given by a B$_{k+1}$-VPG representation. The B$_k$-VPG representation is related to the \textit{edge intersection graphs of paths in a grid} (\textit{EPG-graphs}) introduced by Golumbic \textit{et al.}~\cite{golumbic2009edge}. In an EPG representation, the vertices are represented as 
paths on a grid, and two vertices are adjacent if and only if their corresponding paths share a grid edge.
Pergel and Rz\k{a}\.{z}ewski~\cite{pergel2016edge} proved that is NP-complete to recognize  $2$-bend-EPG-graphs.
 
The study of B$_k$-VPG graphs is motivated by practical applications in 
circuit layouts~\cite{brady1990stretching,molitor1991survey}. 
n the knock-knee layout model, the layout may have multiple layers, and on each layer, the vertex intersection graph of paths on a grid is an independent set. This corresponds to a graph coloring problem, and the minimum coloring problem of VPG-graphs defines the knock-knee multiple layout with minimum number of layers. This model is used by  Asinowski \textit{et al.}~\cite{asinowski2012vertex}, who studied VPG-graphs 
and showed that interval graphs and trees are both subfamilies of B$_0$-VPG, and that  \textit{circle graphs} are contained in the class B$_1$-VPG (where circle graphs are string graphs in which the strings are chords of a circle). Since the problem of coloring a circle graph is NP-complete~\cite{garey1980complexity}, it follows that the coloring problem is also NP-complete for B$_1$-VPG-graphs. Asinowski \textit{et al.}~\cite{asinowski2012vertex} proved that the coloring problem remains NP-complete even for B$_0$-VPG-graphs. 

The class B$_k$-VPG contains all planar graphs, and a central question is how small $k$ can  be. Asinowski \textit{et al.}~\cite{asinowski2012vertex} showed that
every planar graph is a B$_3$-VPG-graph
and Chaplick and Ueckerdt~\cite{chaplick2012planar} showed that every planar graph is a B$_2$-VPG-graph.

In B$_1$-VPG-graphs, four possible \cL-shapes, \cL, \cReverseL, \cGamma~ and \cReverseGamma, are may be used to represent vertices. In an \cL-graph, the vertices are represented with only one of these \cL-shapes. 
Biedl and Derka~\cite{BiedlD15,TBMD16} show that series-parallel graphs, Halin-graphs, and outerplanar graphs are \cL-graphs. 
Felsner \textit{et al.}~\cite{felsner2016intersection} show that every planar $3$-tree is an \cL-graph, and that full subdivisions of planar graphs and line graphs of planar graphs are \cL-graphs. On the other hand, full subdivisions of non-planar graphs are not \cL-graphs~\cite{sinden1966topology}.
    
    In the rest of this paper we restrict our focus to graphs that have an \cL-representation, which we refer to as \cL-graphs. Formally, an $\mathsf{L}$-segment consists of a horizontal and a vertical straight line which form an $\mathsf{L}$.  An {\it $\mathsf{L}$-embedding} is a drawing of $G$ in which each vertex is drawn as an $\mathsf{L}$-segment, and two segments intersect each other if and only if the corresponding vertices are adjacent in the graph. $G$ is an {\it \cLG} if it admits an \cLE. If the corner of each $\mathsf{L}$-segment in an $\mathsf{L}$-embedding lies on a straight line, then it is called  a {\it monotone $\mathsf{L}$-embedding}. A graph is called a {\it \cMLG} if it admits a \cMLE.

	{\bf Our contributions:} We study  \cL-graphs and monotone \cL-graphs and summarize our results as follows:	 
	\begin{itemize}
		\item We introduce a new class of graphs which we call  ``non-jumping graph" and a new vertex labeling which we call ``non-jumping labeling."
		\item We give a full characterization of monotone \cL-graphs by showing that a graph admits a monotone $\mathsf{L}$-embedding if and only if the graph is a non-jumping graph.
		\item 
       We show that given a 
        graph $G$ on $n$ vertices and $m$ edges
        with labeling $\gamma$, there is an $O(n \log n + m) $ time algorithm to determine  whether 
        $\gamma$ is a non-jumping labeling.

		\item We show that outerplanar graphs, convex bipartite graphs, interval graphs, and  complete graphs are subclasses of non-jumping graphs.
		\item We show that distance-hereditary graphs and $k$-leaf power graphs ($k\le 4$) admit $\mathsf{L}$-embeddings.
	 
	\end{itemize}
		
	The rest of the paper is organized as follows. Section~\ref{preliminaries} defines some preliminary graph-theoretic 
	terminology. In Section~\ref{se.nonjumping}, we define a ``non-jumping graph'' and show that  (bull, dart, gem)-free chordal graphs, interval graphs, outerplanar graphs, complete graphs, and  convex bipartite graphs are  non-jumping graphs. We also provide an algorithm to compute a monotone \cL-embedding of a non-jumping graph, and describe some of the properties of non-jumping graphs. In Section~\ref{sec.otherLgraphs}, we show that distance-hereditary graphs and 4-leaf power graphs admit \cL-embeddings. We conclude the  paper with some open problems.

	\section{Preliminaries}
	\label{preliminaries}

	In this section we introduce several definitions. For graph-theoretic definitions not described here, see~\cite{nishizeki2004planar}.
	
		Let $G=(V,E)$ be a graph with a set of vertices $V$  and a set of edges $E$. We say that 
        $G$ is {\it connected} if there is a path between every pair of vertices in $V$. A {\it cycle} of $G$ is a path in which every vertex is reachable from itself. $G$ is a {\it tree} if it does not contain any cycles.
        	$G$ is {\it planar} if it can be embedded in the plane without edge crossings, and {\it outerplanar} if it has a planar drawing in which all vertices of $G$ are placed on the outer face of the drawing. 
              $G$ is {\it bipartite} if its vertices can be partitioned into sets $R$ and $B$  such that every edge connects a vertex in $R$ to one in $B$.  The set of neighbors of $v$ is denoted by $N(v)$.  If a bijective mapping $f:B \rightarrow \{1,2, \hdots , |B|\}$ exists such that for all $r\in R$, and  for any two vertices $x,y\in N(r)$, there does not exist a vertex $z \in B \setminus N(r)$ such that $f(x) < f(z) < f(y)$,  then $G$ is called a {\it convex bipartite} graph.  
            $G$ is called an {\it interval graph}, and a set of intervals $S$ is called an {\it interval representation} of $G$, if there exists a one-to-one correspondence between vertices of $G$ and intervals in $S$, such that $u$ and $v$ are adjacent in $G$, if and only if, their corresponding  intervals intersect. 
        
		Let $G'=(V',E')$ be a graph such that $V'\subseteq V$ and $E'\subseteq E$. Then $G'$ is called a {\it  subgraph} of $G$.	The subgraph $G'$ is an {\it induced subgraph} of $G$ if $E'$ consists of all  the edges in $E$ that have both endpoints in $V'$.  $G$ is a  {\it distance-hereditary graph} if and only if for every pair of vertices $u$, $v$ all induced path between $u$ and $v$ have the same length. 
		
		 $G$ is a {\it $k$-leaf power} graph if there is a tree $T$ whose leaves correspond to the vertices of $G$ in such a way that two vertices are adjacent in $G$ precisely when their distance in $T$ is at most $k$.
		 We say that $G$ is a {\it leaf power graph} if it is a $k$-leaf power for some $k$.
         
         A vertex $u\in V$ is a \textit{pendant} vertex if it has degree 1. For two  vertices $u, v \in V$, if $u$ and $v$ are neighbors and $N(u) \setminus \{v\} = N(v) \setminus\{u\}$, then $u$ and $v$ are called {\it true twins}.  If $u$ and $v$ are not neighbors and  $N(u)=N(v)$, we say that $u$ and $v$ are {\it false twins}.
		
	   A vertex $v$ is called {\it simplicial} in  $G$ if the subgraph of $G$ induced by the vertex set $\{v\} \cup N(v)$ is a complete graph. An ordering of $\{v_1, v_2, \ldots, v_n\}$ is a {\it perfect elimination ordering} of $G$ if each $v_i$ is simplicial in the subgraph induced by the vertices $\{v_1, v_2, \ldots, v_i\}$. $G$ is a {\it  chordal} graph if it has a perfect elimination ordering.
	
			An {\it $\mathsf{L}$-segment } consists of a horizontal and a vertical straight-line segment which together look exactly like an $\mathsf{L}$, with no  rotation. 
	 Let $\mathbb{L}$ be  an  $\mathsf{L}$-embedding of $G$ and  $v$ be a vertex of $G$. We denote  the corresponding  $\mathsf{L}$-segment of $v$ in $\mathbb{L}$  by \cL($v$). The $\mathsf{L}$-segment  is defined by its corner position, the height of its vertical and the width of its horizontal line segments, denoted by  $(v.x, v.y)$, $v.h$, and $v.w$, respectively. Let \cL($u$) and \cL($v$) be two $\mathsf{L}$-segments in $\mathbb{L}$. The segments  \cL($u$) and \cL($v$) might cross each other multiple times in case of overlapping horizontal or vertical segments. 
     In this paper we consider only \cL-embeddings with single crossings, so that if $(u,v)\in E$ then either the horizontal segment of  \cL($v$) crosses the  vertical segment of  \cL($u$), or the vertical segment of  \cL($v$) crosses  the horizontal  segment of  \cL($u$).
    
    \section{Non-jumping graphs }
	\label{se.nonjumping}
	In this section we give a formal definition of a non-jumping graph. Then we show that several classes of graphs are non-jumping graphs.   Before we define non-jumping graphs, we must first define a non-jumping labeling of a graph $G=(V,E)$.
	A {\it non-jumping labeling} of $G$ is a vertex labeling  $v_1, v_2, \ldots, v_n$ such that if $(v_i,v_k),(v_j, v_l)$ $\in E$ and $i<j<k<l$, then $(v_j,v_k)\in E$. Figure~\ref{fig:njumpingex}(a) provides an example of a non-jumping labeling. If $G$ admits a non-jumping labeling, then we say that $G$ is a \textit{non-jumping graph}; if $G$ has no non-jumping labeling then $G$ is called a {\it jumping graph}.  If a vertex labeling contains a vertex $v_j$ such that $(v_i,v_k),(v_j, v_l) \in E$ but $(v_j,v_k) \notin E$ (where $i<j<k<l$), then $v_j$ is called a {\it jumping vertex} for $v_i, v_k$, and $v_l$. For example, the vertex $v_3$ is a jumping vertex in the graph shown in Fig.~\ref{fig:njumpingex}(c). Clearly, a non-jumping labeling does not contain any jumping vertex. 
\subsection{Families of non-jumping graphs}

	\begin{figure}[t]
		\hfil \includegraphics[width=1\textwidth]{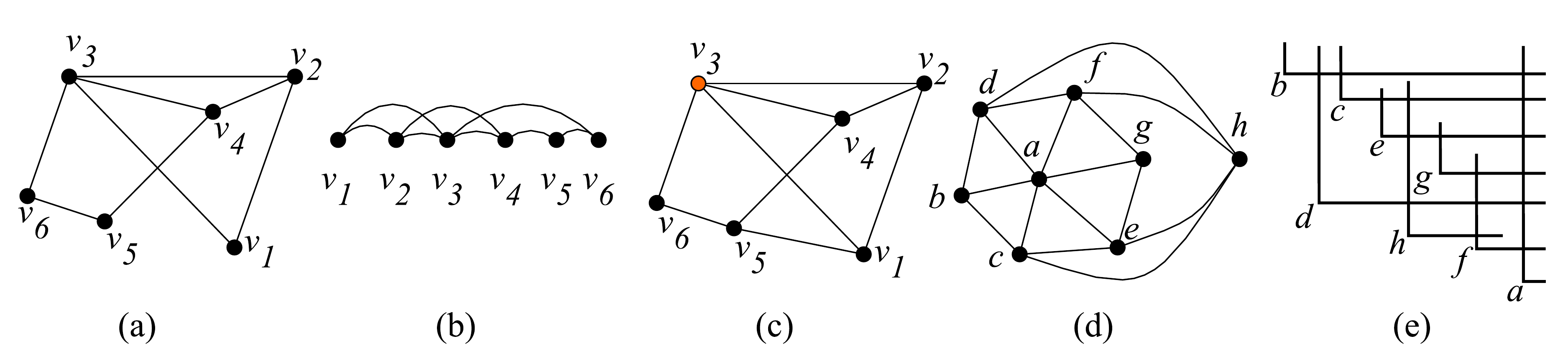} \hfil
		\vspace{-.8cm}\caption{(a) A non-jumping labeling of a graph $G$,  (b) an ordering $\gamma=\{v_1,v_2, v_3, v_4, v_5\}$ of $G$, (c) a jumping vertex $v_3$,  (d) a jumping graph $G'$, and (e) an \cLE  of $G'$. }
		\label{fig:njumpingex}
	\end{figure}

	One can easily verify that paths, cycles, and complete graphs are non-jumping graphs. In this section, we describe several other types of graphs can be classified as non-jumping graphs.  We begin with outerplanar graphs.
    \begin{theorem}
		\label{outerplanar}
		Let $G$ be an outerplanar graph. Then $G$ is a non-jumping graph, and a non-jumping labeling of $G$ can be found in linear time. 
    \end{theorem}
    \begin{proof}

    Every outerplanar graph admits a one page book embedding~\cite{BERNHART1979320} which can be found in linear time.
    In a one page book embedding of a graph, we place each vertex of the graph on the spine of the book and each edge can be drawn on one page without edge crossing.  If we consider the sequence of vertices as a labeling of a one page book embedding, there is no jumping vertex because there is no pair of edges $(v_i,v_k),(v_j, v_l) \in E$ where $i<j<k<l$. \qed
	  \end{proof}
	
	We next show that (bull, dart, gem)-free chordal graphs are non-jumping graphs. The bull, dart and gem are shown in the Fig.~\ref{fig:3-leaf-power-example}(a).  Before the proof, we define a few terms as follows: Let $T$ be a tree, and $v$ be a vertex  of $T$. We denote the subtree of $T$ rooted at $v$ by $T_v$. We denote the parent of $v$ by $v'$,  and  parent of  $v'$ by  $v''$. A vertex $u$ is said to be  an {\it uncle} of $v$ if $u'=v''$.

	\begin{theorem}
		\label{non-jumping_3-leaf}
		Every (bull, dart, gem)-free chordal graph is a non-jumping graph.
	\end{theorem}

\begin{proof}
	Let $G=(V,E)$ be a (bull, dart, gem)-free chordal graph of $n$ vertices. Then there is a tree $T$ whose leaves correspond to the vertices of $G$ such that two vertices are adjacent in $G$ precisely when their distance in $T$ is at most three~\cite{Rautenbach20061456}; see Fig.~\ref{fig:3-leaf-power-example}(b-c). Hence $G$ is a 3-leaf power graph of $T$. 	We use the notation $\overline{v}$ to indicate that the leaf $v$ of $T$ corresponds to the vertex $\overline{v}$ in $G$.    Let $\overline{u}$ and $\overline{v}$ be vertices of $G$.
    
  Since $G$ is a 3-leaf power graph of $T$, $(\overline{u}, \overline{v})\in E$ if and only if $u$ and $v$ are siblings, or $u$ is an uncle of $v$, or   $v$ is an uncle of $u$. We find an ordering of vertices of $G$ using $T$ as follows:
	 We first root $T$ at a non-leaf vertex $x$ of $T$. 
		We then sort each subtree of $T$ rooted at each vertex in counterclockwise order, according to depth in ascending order.  

	\begin{figure}[t]
	\hfil \includegraphics[width=0.9\textwidth]{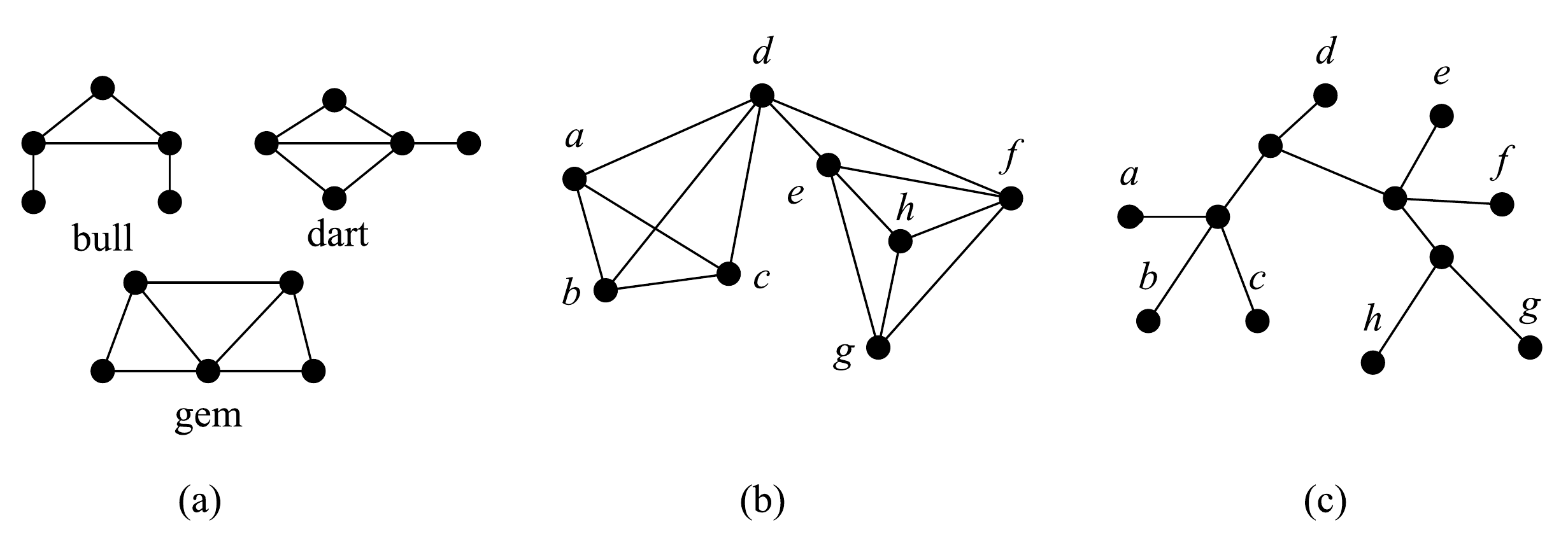} \hfil
	\vspace{-.6cm}\caption{(a) A bull, a dart, a gem. (b) A 3-leaf power graph $G$, and (c) the 3-leaf power tree of $G$.}
	\label{fig:3-leaf-power-example}
\end{figure}
	 Let $\gamma'=\{v_1, v_2, v_3, \ldots, v_n\}$ be the ordering of the leaves taken from the counterclockwise DFS traversal on $T$ starting from $x$.  We now prove that  $\gamma=\{\overline{v_1}, \overline{v_2}, \overline{v_3}$, $\ldots, \overline{v_n}\}$ is a non-jumping labeling of $G$ by supposing that $(\overline{v_i}, \overline{v_k}), (\overline{v_j}, \overline{v_l})\in E$ with $i<j<k<l$, and showing that $(\overline{v_j}, \overline{v_k})\in E$.

	 It is easy to see that $(\overline{v_i}, \overline{v_k})\in E$ if and only if  $v_i$ and  $v_k$ are siblings, or $v_i$ is an uncle of $v_k$. Since we sorted the vertices by their depth before taking the ordering, and since $i<k$,  $v_k$ can not be an uncle of $v_i$. Thus, we have two cases to consider:
	 
	 {\bf Case 1:   $v_i$ and  $v_k$ are siblings.} Since the order was taken from DFS traversal, $v_i, v_{i+1}, \ldots, v_k$ are siblings. So, $v_j$ and $v_k$ are siblings, and we have $(\overline{v_j}, \overline{v_k})\in E$. 
	 
	 {\bf Case 2: $v_i$ is an uncle of $v_k$.}   If $v_i$ and $v_j$ are siblings, then  $(\overline{v_j}, \overline{v_k})\in E$, because the distance between $v_i$ and $v_k$ and the distance between $v_j$ and $v_k$ are the same.  Otherwise, we can prove that $v_j$ and $v_k$ are siblings. Suppose that  $v_j$ and $v_k$ are not siblings. Then $v_j\in T_o$, where $o$ is a non-leaf child of $v'_i$  that was encountered before $v'_k$ in the traversal. Thus, the path between $v_j$ and $v_l$ contains $v'_i$ due to the ordering of the vertices and the positions of $i, j, k,$ and $l$; see Fig.~\ref{fig:3leaf-non-jumping-proof}. 
     Now, the distance between $v'_i$  and $v_j$ is at least 2, and  the distance between $v'_i$  and $v_l$ is at least 2.  This means that $(\overline{v_j}, \overline{v_l})\notin E$, which is not true. The contradiction shows that $v_j$ and $v_k$ must be siblings and  $(\overline{v_j}, \overline{v_k})\in E$. \qed
	\end{proof}

We now show that every interval graph has a non-jumping labeling.
   
  	\begin{theorem}
		\label{non-jumping_interval}
		Every interval graph is a non-jumping graph.\vspace{-.6cm}
	\end{theorem}
    	\begin{figure}[!htbp]
		\hfil \includegraphics[width=1\textwidth]{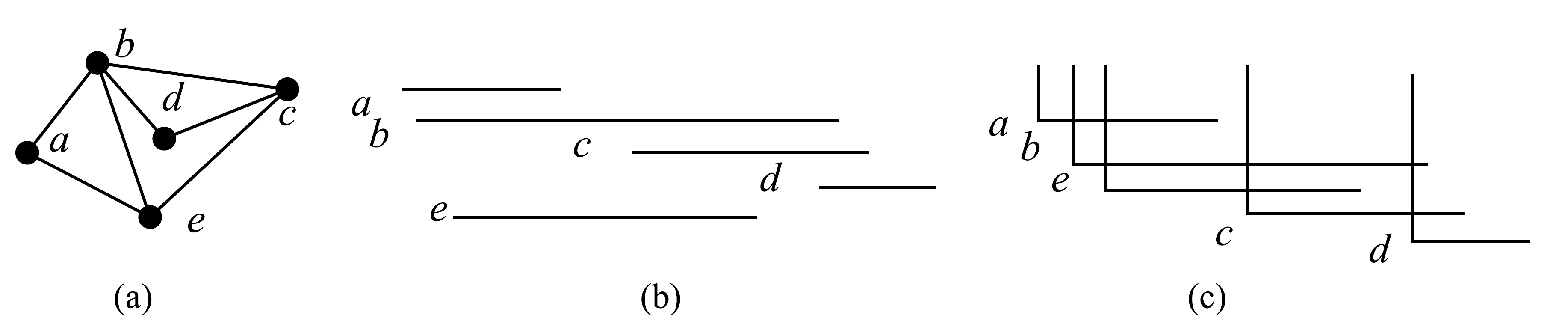} \hfil
		\vspace{-.8cm}\caption{(a) A graph $G$, (b) an interval representation of $G$, and (c) an \cLE of $G$.}
		\label{fig:intervalgraph}
	\end{figure}
	\begin{proof}
		Let $G$ be an interval graph of $n$ vertices. Let $S=\{|a_1,b_1|, |a_2,b_2|, \ldots, |a_n,b_n| \}$ be an interval representation of $G$. We denote the endpoints $a_i$ and $b_i$ of the interval corresponding to $v_i$ by $v_i(a)$ and  $v_i(b)$, respectively. Note that $v_i(a)<v_i(b)$; see Fig.~\ref{fig:intervalgraph}. Let  $\gamma=\{v_1, v_2,\ldots, v_n\}$ be an ordering of vertices of $G$  in non-decreasing order of $v_i(a)$ $(1\le i \le n)$. If $i<j$ then $v_i(a)< v_j(a)$. We now prove that $\gamma$ is a non-jumping labeling of $G$.  By way of contradiction, assume that $\gamma$ is a jumping labeling. Then $\gamma$ contains a jumping vertex $v_j$. By definition, there exists edges $(v_i, v_k),(v_j, v_l) \in E$ with $i<j<k<l$ such that $(v_j,v_k)\notin E$  This means  $v_i(b)> v_k(a)$ and $v_j(b)> v_l(a)$. By construction,   $v_j(b)>v_k(a)$, since  $v_k(a)<v_l(a)$ and  $v_l(a) < v_j(b)$.  Since  $v_j(b)>v_k(a)$ and $v_j(a)<v_k(a)$, we have $(v_j,v_k)\in E$, a contradiction. 
		\qed
	\end{proof}

	\begin{theorem}
		\label{non-jumping_convex}
		Every convex bipartite graph is a non-jumping graph.
	\end{theorem}
	
	\begin{proof}

		Let $G = (R\, \cup\, B, E)$ be a convex bipartite graph with $V(G) = R\, \cup\, B$, where $R\, \cap\, B = \emptyset$. Without loss of generality, suppose that $G$ is convex over $B$. Then there exists a bijective mapping $f:B \rightarrow \{1,2, \hdots ,|B|\}$ such that for all $v \in R$ and any two vertices $x,y \in N(v)$, there is no vertex $z \in B \setminus N(v)$ such that $f(x) < f(z) < f(y)$. 

{\makeatletter
\let\par\@@par
\par\parshape0
\everypar{} 
\begin{wrapfigure}{l}{0.6\textwidth} 
\begin{center}
\hfil \includegraphics[width=0.6\textwidth]{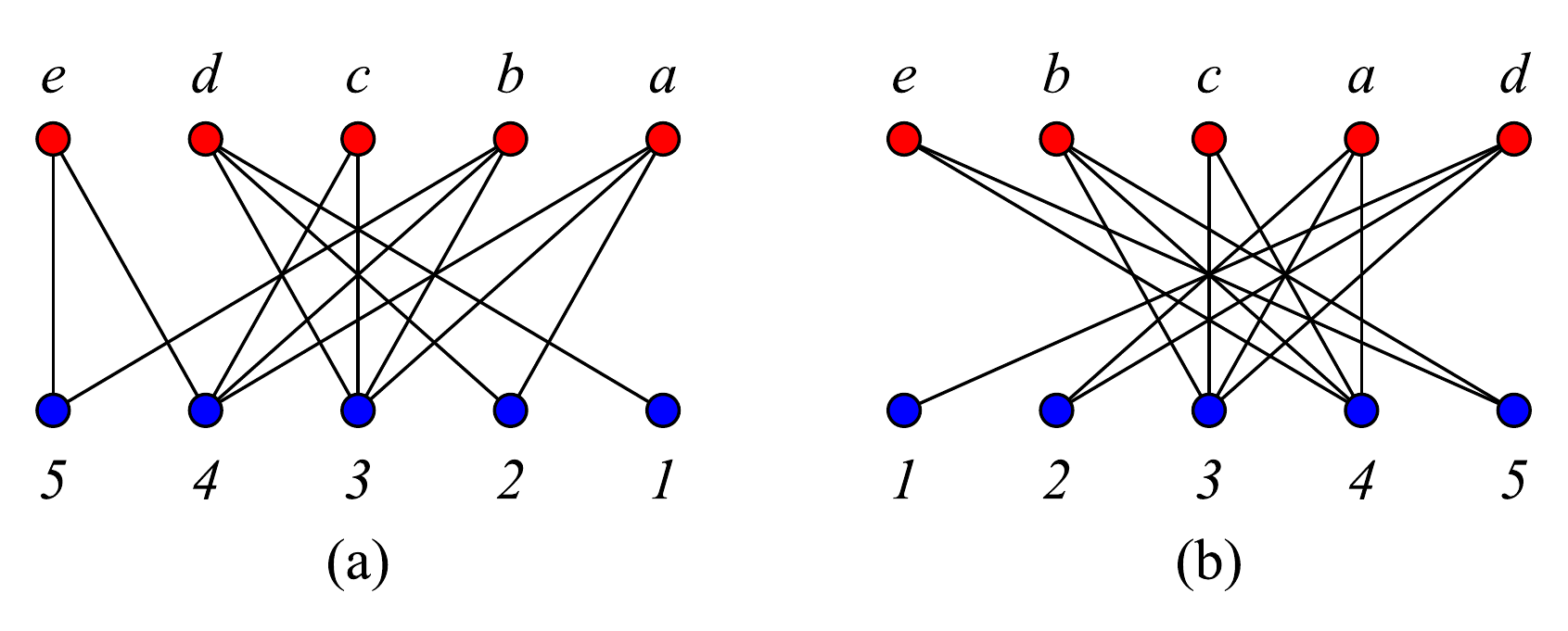}  \hfil 
\vspace{-.5cm}\caption{(a) A convex bipartite graph $G$ and (b) a vertex ordering of $G$.}
\label{fig:convex_proof}
\end{center}
\end{wrapfigure}

We define $s:R\rightarrow B$  so that for any vertex $v \in R$, $s(v) = min_{b \in N(v)}f(b)$. Suppose we sort the vertices $r \in R$ in non-increasing order of $s(r)$. Let $R_{sort}$ be the new ordering, and let $B_f$ be the vertices $b \in B$ sorted in increasing order of $f(b)$. For example, in Fig.~\ref{fig:convex_proof}(b), $R_{sort}=\{e, b, c, a, d\}$ and $B_f=\{1, 2, 3, 4, 5\}$. 
        We now prove that $\gamma = \{R_{sort}, B_f\}$  is a non-jumping labeling of $V(G)$.
\par}%

		For $\gamma$ to be a non-jumping labeling, it must be true that for all positions $i<j<k<l$ in $\gamma$ , if $(v_i,v_k) \in E$ and $(v_j,v_l) \in E$, then $(v_j,v_k) \in E$. Consider such a pair of edges $(v_i, v_k), (v_j,v_l) \in E$. Since $i<k$ and $G$ is a bipartite graph, $v_i \in R$ and $v_k \in B$. Similarly, $v_j \in R$ and $v_l \in B$. 

		Because $i<j$ and the vertices in $R_{sort}$ were ordered in non-increasing order of $s(r)$, we have $s(v_i) \geq s(v_j)$. Also, since $v_k \in N(v_i)$, we know that $f(v_k) \geq s(v_i)$
        Consequently, $f(v_k) \geq s(v_i) \geq s(v_j)$. 

		Since $(v_j, v_l) \in E$ and $G$ is convex on $B$, $N(v_j)$ must contain all vertices $v \in B$ whose mapping $f(v)$ lies in the interval $[s(v_j), f(v_l)]$. Because $f(v_k) \geq s(v_j)$ and, from the ordering in $B_f$, $f(v_k) < f(v_l)$, we find that $f(v_k)$ lies in the interval $[s(v_j), f(v_l)]$. Thus, $v_k \in N(v_j)$, which means that $(v_k,v_j)\in E$, as required.
		\qed
	\end{proof}

\subsection{Characterization of non-jumping graphs}
The graph shown in the Fig.~\ref{fig:njumpingex}(d) is an example of jumping graph. In Theorem~\ref{jumping_graph}, we prove that there is no non-jumping labeling for this graph. Due to space limitations, the proof of Theorem~\ref{jumping_graph} is given in the Appendix. 
	\begin{theorem}
		\label{jumping_graph}
	Not all graphs are non-jumping graphs.
	\end{theorem}

Recall that a  monotone  \cL-embedding is an  \cL-embedding such that the corners of each \cL-segment  are on a straight line. We can completely characterize monotone  \cL-graphs in terms of non-jumping graphs.

\begin{theorem}
	\label{non-jumping_theorem_iif}
	A graph $G$ admits a monotone \cL-embedding if and only if $G$ is a non-jumping graph.
\end{theorem}

We prove Theorem~\ref{non-jumping_theorem_iif} by first showing that any non-jumping graph $G$  admits a monotone \cL-embedding in Lemma~\ref{non-jumping_drawing}.  The converse is proven in Lemma~\ref{non-jumping_iif}.

	Before we begin, we note that if a graph has a monotone  \cL-embedding with the corners of the \cL's on a line that is drawn vertically or horizontally, then for any pair
	of vertices $(v_i,v_j)$, there can only be an edge $(v_i,v_j)$ if $i+1=j$, i.e., the graph is a subgraph of a path. 
	Thus, the graph is trivially non-jumping, as
	there cannot be any indices $i<j<k<l$ in a labeling $\gamma$ such that $(v_i,v_k) \in E$ or $(v_j,v_l) \in E$.  A graph with no edges is also trivially non-jumping, and admits a ``degenerate'' \cMLE in which no \cL~ would intersect another even if their arms were extended indefinitely.

 \begin{wrapfigure}[16]{l}{0.4\textwidth}
  \begin{center}
\includegraphics[width=0.35\textwidth]{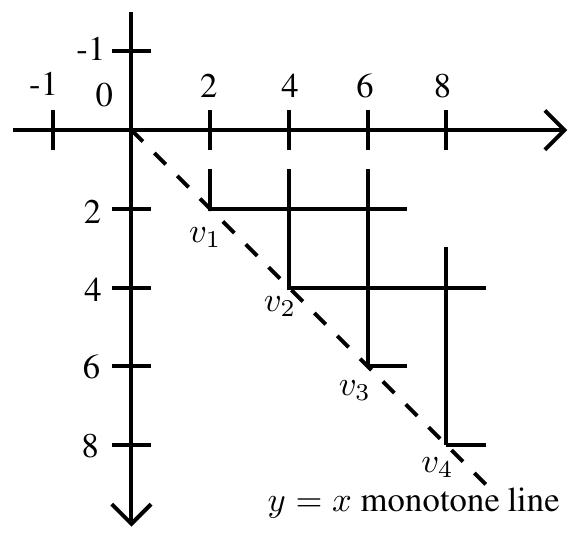}  
\caption{Monotone positioning of the \cL-segments corresponding to $v_1, \ldots, v_4$ on the line $y=x$. }
\label{fig:monotonepositioning}
 \end{center}
\end{wrapfigure}
For convenience, we define a coordinate system over the quarter-plane $\mathbb{R}^2$ beginning with  $(0,0)$ in the top-left corner, 
	and $x\mbox{-}$ and $y\mbox{-}$coordinates increasing to the right and downward respectively. This choice of coordinate system will allow us to construct a \cMLE so that the corner of every \cL lies on the line $y=x$. Moreover, any non-trivial
    \footnote{By non-trivial, we mean a \cMLE with at least one intersection, and with \cL's that are not aligned horizontally or vertically.} 
    \cMLE can be expanded (see Lemma~\ref{lemma:expanding}), translated, and rescaled to create an equivalent embedding with the corners of each \cL on this line.  Note that once we have a drawing with the corners of each \cL line arranged on  $y=x$, we can perform arbitrary affine transformations on the coordinate system without rotating any of the \cL's themselves. For the rest of the paper, unless otherwise indicated, every \cMLE will have its corners aligned on the line $y=x$ in this way.

	\begin{lemma}
		\label{non-jumping_drawing}
		Let $G$ be a non-jumping graph of $n$ vertices and $m$ edges. Then $G$  admits \cMLE on a grid of size $O(n)\times O(n)$, and this embedding can be computed in $O(n+m)$ time. 
	\end{lemma}
	
	\begin{proof}
	Let $G$ be a non-jumping graph of $n$ vertices and $\gamma=\{v_1, v_2, \ldots, v_n\}$ be a non-jumping labeling of $G$. 
		 If there are edges $(v_i,v_k), (v_j,v_l) \in E$ such that $i<j<k<l$, then $(v_j,v_k) \in E $
				
		We now construct an \cL-monotone drawing for $G$ using the coordinate system given above.  Let \cL($v$) be the \cL-drawing of vertex $v$. Then $(v.x,v.y)$ is the corner of \cL($v$) and the horizontal and vertical arms of \cL($v$) have lengths $v.w$ and $v.h$ respectively.
        
		For each $v_j \in V$, let $v_j.x = v_j.y = 2j$; this places all corners on the line $y = x$.   Also, for each $v_j$, if there exists an index $i<j$
		such that some $(v_i,v_j) \in E$, then for $a =$ min$\{i\,|\, i < j$ and $(v_i,v_j) \in E\} $, define $v_j.h = 2|j-a|+1$.  If there is no
		such index $i$, then let $v_j.h = 1$.  Similarly, if there is some index $k>j$ such that $(v_j,v_k) \in E$, 
		then let $b=$ max$ \{k\, |\, k > j$ and $(v_j,v_k) \in E\}$
		and define $v_j.w = 2|j-b|+1$; otherwise, let $v_j.w = 1$; see Fig.~\ref{fig:monotonepositioning}. 

		To see that this is a valid \cL-monotone drawing, first recall that the corners of all the \cL's are on the diagonal line
		$y=x$.  Also note that for each index $a<b$, $v_a.x < v_b.x$ and $v_a.y < v_b.y$.	
        We must show that for indices $a$ and $b$ with $(v_a,v_b) \in E$, \cL($v_a$) and \cL($v_b$) intersect. Without loss of generality, suppose $a<b$.  Then, we must show that:
		$
			|v_a.x-v_b.x| < v_a.w$
		  and  
			$|v_a.y-v_b.y| < v_b.h;		$           
		Finally, we must show that for $(v_j,v_k) \notin E$ with $j<k$, either
		$|v_j.x-v_k.x| > v_j.w$
		 or 
		$|v_j.y-v_k.y| > v_k.h$.

		It is clear from the computation of the width and height of each \cL\  that whenever $(v_a,v_b)\in E$,
		\cL($v_a$) and \cL($v_b$) intersect as described above.
        
		Now let $v_j$ and $v_k$ be vertices in $G$ such that $(v_j,v_k) \notin E$, with $j<k$. 
		Suppose for a contradiction that \cL($v_j$) and \cL($v_k$) intersect; 
		that is, $|v_j.x-v_k.x| < v_j.w$ and $|v_j.y-v_k.y| < v_k.h$.
		Since $v_j$ and $v_k$ are not adjacent, by the construction of $v_j.w$, there must be some vertex $v_l$ such that $l>k$ and
		\[
		v_j.w = |v_j.x - v_l.x| + 1 > |v_j.x-v_k.x| \\ 
        \]
        Similarly, by the construction of $v_l.h$, there must be some vertex $v_i$ such that $i<j$ and
        \[
		v_k.h = |v_k.y - v_i.y| + 1 > |v_j.y-v_k.y|
		\]
		We now have $i < j < k < l$ and the two edges $(v_i,v_k), (v_j,v_l)\in E$.  Since the indices $i,j,k,l$ are taken from a non-jumping labeling
		of $G$, we must have $(v_j,v_k) \in E$, a contradiction.
        
          The entire drawing is contained in a rectangle of dimensions $2n\times 2n$.  To see this, note that no corner of any \cL-segment will be placed to the left of the line $x=2$, nor below the line $y=2n$.  Also, no horizontal arm of an \cL\ will extend to the right beyond the line $x=2n+1$, as this is one unit to the right of \cL($v_n$), nor will any vertical arm extend above the line $y=1$. 
        
       We can construct this drawing in $O(|V| + |E|)$ time. First, for each $ v \in V$, we plot the corner of \cL($v$) at $(v.x,v.y)$, and draw its two arms with unit length. Then, for each edge $(v_i,v_j) \in E$ with $i<j$, we extend the horizontal arm of \cL($v_i$) to have length at least $2|i-j|+1$, and extend the vertical arm of \cL($v_j$) to have length at least $2|i-j|+1$.
		\qed
	\end{proof}
	
	\begin{lemma}
		\label{non-jumping_iif}
	Let $\mathbb{L}$ be a monotone \cL-embedding of a graph $G$. Then $G$ is a non-jumping graph.
	\end{lemma}
	
	\begin{proof}
	 
		 Since $\mathbb{L}$ is monotone, the corners of each \cL($v$) lie on a straight line. 	   
		 Let $\gamma({\mathbb{L}})=\{\mathsf{L}(v_1), \mathsf{L}(v_2), \ldots, \mathsf{L}(v_n)\}$  
		 be an ordering of the \cL-segments according to their corner positions from left to right. 
		 If the line on which the corners of the \cL's lie is horizontal or vertical, then as described above, the graph is a subgraph
		 of a path, and is trivially non-jumping.  Similarly, if the corners lie on a line with negative slope, then there are no edges,
		 so the graph is trivially non-jumping.

         The remaining possibility is that the corners of the \cL-segments lie on a line with positive slope. In this case, for each pair of indices $a<b$, we have
		 $v_a.x < v_b.x$ and $v_a.y < v_a.y$.   
		 Now, $\gamma({\mathbb{L}})$ gives us an ordering $\gamma=\{v_1, v_2, \ldots, v_n\}$ of the corresponding  vertices  of $G$.
   	 We want to prove that $\gamma$ is a non-jumping labeling of $G$. 
		 For any four vertices $v_i,v_j,v_k,v_l$ with $i<j<k<l$, 
		 we must show that if \cL($v_i$) intersects \cL($v_k$) and \cL($v_j$) intersects  \cL($v_l$),
		 then \cL($v_j$)   and  \cL($v_k$) also intersect. 

To begin, note that if  \cL($v_i$)  and  \cL($v_k$) intersect,  then 		
			$|v_i.x-v_k.x| < v_i.w$    and 
			$|v_i.y-v_k.y| < v_k.h$.		
		Similarly, if  \cL($v_j$)   and  \cL($v_l$) intersect,  we have
		$		|v_j.x-v_l.x| < v_j.w$	  and  	$|v_j.y-v_l.y| < v_l.h.$ 	
		By the ordering of $\gamma(\mathbb{L})$, we have $v_i.x < v_j.x < v_k.x$ and $v_j.y < v_k.y < v_l.y$.  Thus,		 
			$|v_j.x-v_k.x| < |v_j.x-v_l.x| < v_j.w$
		~ and ~
			$|v_j.y-v_k.y| < |v_i.y-v_k.y| < v_k.h$.  So \cL($v_j$) and \cL($v_k$) intersect, and $(v_j,v_k) \in E$. 
        \qed
	\end{proof}

In the proof of Theorem~\ref{jumping_graph},	we show that the graph in Fig.~\ref{fig:njumpingex}(d) is a jumping graph. However, it is easy to verify that the \cLE in Fig.~\ref{fig:njumpingex}(e) is the \cLE of the jumping graph. In Theorem~\ref{non-jumping_theorem_iif}, we showed that a graph $G$ admits a monotone \cL-embedding if and only if $G$ is a non-jumping graph. This proves the following theorem:
	\begin{theorem}
		\label{jumping_graph_L}
	Not all \cL -graphs are monotone \cL -graphs.
	\end{theorem}

		\subsection{Recognition of a non-jumping labeling}
	While it is difficult to determine whether a particular graph is non-jumping, the following theorem shows that we can easily verify whether a given labeling for a graph is a non-jumping labeling.
	
	\begin{theorem}
		Given a graph $G=(V,E)$ with vertex labeling $\gamma=\{v_1, v_2,\ldots, v_n\}$, it can be determined in $O(|V|\log |V| + |E|)$ time whether $\gamma$ is a non-jumping labeling for $G$.

	\end{theorem}
	
		\begin{proof}
		Using the procedure described in Lemma \ref{non-jumping_drawing}, we can construct an \cL-monotone embedding of a graph $G=(V,E)$ in $O(|V| + |E|)$ time given a non-jumping labeling $\gamma$. Let us call the procedure $P$.

		From Lemma \ref{non-jumping_iif}, we know that given any \cL-monotone embedding, we can construct a non-jumping labeling $\gamma$ by sorting the vertices $v_i \in G(V)$ in increasing order of this corner coordinate $v_{i}.x$.
        Let us call this order $V_{sort}$. The $\gamma$ thus constructed from $V_{sort}$ is a non-jumping labeling.

		$P$ produces a valid \cL-monotone embedding if and only if the input labeling $\gamma$ is non-jumping. To prove this, let us suppose we get a valid \cL-monotone embedding from $P$ using a jumping labeling $\gamma_{jump}$. Let us arrange the vertices $v_i$ in increasing order of corner coordinates $v_{i}.x$ to obtain $V_{sort}$. $V_{sort}$ must give a non-jumping labeling. Thus, our assumption that $\gamma_{jump}$ is a jumping labeling is invalid and we get a contradiction.

		We can use this to test if any $\gamma$ is non-jumping or not. Let the drawing produced by $P$ using $\gamma$ be $P(\gamma)$. If $P(\gamma)$ is a valid \cL-monotone embedding, $\gamma$ must be non-jumping. A valid \cL-embedding has $2|V|$ line segments (one vertical and one horizontal for each \cL-shape). Similarly, there are $|E| + |V|$ line segment intersections (one intersection for every $(v_i,v_j) \in E$, and one intersection at the corner of each \cL-shape).

		Using an orthogonal line segment intersection search (e.g., a sweep line algorithm as described in ~\cite{chazelle1994algorithms}), all intersections in $P(\gamma)$ can be listed in $O(N\log N + k)$ time, where $N = 2|V|$ is the number of total line segments, and $k = O(N^2)$ is the number of possible intersections. It suffices to check the first $|E| + |V|$ intersections to determine if $P(\gamma)$ is a valid \cL-monotone embedding: if there is an unwanted intersection in the first $|E| + |V|$ intersections, then the embedding is invalid, and $\gamma$ is jumping. On the other hand, if there are more than $|E|+|V|$ intersections, these additional intersections must be invalid and $\gamma$ is jumping. Otherwise, $\gamma$ is non-jumping.  
        
        Since only the first $k = |E| + |V|$ intersections need to be examined, we only need $O(|V|\log |V| + |E|)$ time to determine if a labeling is non-jumping or not.\qed
        
	\end{proof}

	\section{Other $\mathsf{L}$-graphs}
	\label{sec.otherLgraphs}
	In this section we prove that distance-hereditary graphs and $k$-leaf power graphs (for $k\le 4$) admit  \cL-embeddings. We begin with a lemma about transformations of \cL-embeddings; this result will allow us to derive new \cL-graphs from old ones.  

\begin{lemma}
    \label{lemma:expanding}
    Let $\LLL$ be an \cL-embedding of a graph $G$, and \cL($v$) be the \cL~ corresponding to vertex $v$. Then,
    a valid \cL-embedding $\LLL'$ of $G$ can be constructed from $\LLL$ by expanding an infinitesimal slice of $\LLL'$ that is parallel to an arm of \cL($v$).
    
    \end{lemma}
    \begin{proof}

    There are four ways in which $\LLL$ can be expanded:

      \begin{enumerate}
		\item[a)] Expand $\LLL$ rightward with respect to \cL($v$)  as follows: For every \cL($w$) in $\LLL$ with a corner to the right of \cL($v$), move \cL($w$) to the right by one unit; also, for every vertex \cL($u$) with a corner to the left of or vertically aligned with \cL($u$) that intersects such an \cL($w$), extend the horizontal arm of \cL($u$) to the right by one unit.
        \item[b)] Expand $\LLL$ leftward with respect to \cL($v$) as follows: For every \cL($u$) in $\LLL$ with a corner to the left of \cL($v$), move \cL($u$) to the left by one unit; also, if such an \cL($u$) intersects a \cL($w$) that has its corner vertically aligned with or to the right of \cL($v$), then extend the horizontal arm of \cL($u$) by one unit.
        \item[c)] Expand $\LLL$ upward with respect to \cL($v$), by replacing the words `right' and `left' in the description of expanding rightward above with `up' and `down' respectively, and exchanging the words `horizontal' and `vertical'.
        \item[d)] Expand $\LLL$ downward by similarly modifying the description of a leftwards expansion.
    \end{enumerate}
        To show that any of these operations produces another valid embedding $\LLL'$ of $G$, we can simply observe that in each transformation, all intersections of \cL's are preserved, and no new intersections are introduced.
        \qed
    \end{proof}
    
    Using this result, we find that certain modifications of an \cL-graph result in another \cL-graph.


\begin{theorem}
		\label{theorem:distance-hereditary}
		Let $G$ be a graph that admits  an \cL-embedding, and $G'$ be a graph constructed from $G$ by adding a pendant vertex, a true twin, or a false twin in $G$. Then $G'$ admits  an \cL-embedding.
	\end{theorem}

	\begin{proof}
	 Let $\mathbb{L}$ be an \cL-embedding of $G$, 
     and suppose we derive $G'$ from $G$ by adding a pendant vertex $v$ to $G$ with neighbor $u$.  
     Let \cL($u$) represent $u$ in $\LLL$.
     To create $\LLL'$, we must place \cL($v$) so that it intersects with \cL($u$) and no other \cL.  
     To be sure there is room to do so, we first expand $\LLL$ rightward two units, and both upward and downward one unit, with respect to \cL($u$). We then place \cL($v$) with its corner one unit to the right and one unit below the corner of \cL($u$), giving it horizontal arm length 1 and vertical arm length 2.

Suppose instead that we derive $G'$ from $G$ by replacing a vertex $u$ with true twin vertices $v$ and $w$ so that $v$ and $w$ are adjacent to all the neighbors of $u$, and are also adjacent to one another. We construct $\LLL'$ representing $G'$ as follows.  Replace \cL($u$) with \cL($v$), so that \cL($v$) retains all the intersections of \cL($u$).  Now expand the drawing both rightward and downward one unit with respect to \cL($v$) to create room for \cL($w$).  We give \cL($w$) a vertical arm length that is one greater than that of \cL($v$), and a horizontal arm length one less than that of \cL($v$) after the rightward expansion, and place its corner one unit down and to the right of the corner of \cL($v$).  Thus,  \cL($v$) and \cL($w$) each intersect every \cL-segment that \cL($u$) intersected, and also intersect each other. 	
 
	If we construct $G'$ from $G$ by replacing a vertex $u$ with false twin vertices $v$ and $w$, we can proceed similarly. We first expand $\LLL$ leftward and downward one unit with respect to \cL($v$). Next, we place \cL($w$) one unit down and to the left of \cL($v$), and give \cL($w$) vertical and horizontal arm lengths one unit greater than those of \cL($v$).   Now \cL($w$) intersects the same \cL-segments that \cL($v$) intersects, but does not intersect \cL($v$) itself.		 
		\qed
	\end{proof}
	 
			If $G$ is a distance-hereditary graph, then $G$ can be built up from a single vertex by a sequence of the following three operations: 
        a) add a pendant vertex, b)
         replace any vertex with a pair of false twins, and c)
        replace any vertex with a pair of true twins. ~\cite{BANDELT1986182}  Thus, Theorem~\ref{theorem:distance-hereditary} immediately yields the following corollary.
	
	\begin{corollary}
		\label{dsh}
		Let $G$ be a distance-hereditary graph. Then $G$ is an  $\mathsf{L}$-graph.
	\end{corollary}

It is easy to see that $1$-leaf power graphs and $2$-leaf power graphs are \cL -graphs. We can use Theorem~\ref{non-jumping_3-leaf} and Lemma~\ref{non-jumping_drawing} to also show that $3$-leaf power graphs are   \cL -graphs. The proof of the following theorem on $4$-leaf power graphs is given in the Appendix.

\begin{theorem}
	\label{theorem:4-leaf-power}
 
Every $4$-leaf power graph admits an \cL-embedding. 
	
\end{theorem}

	\section{Conclusions and Future Work}
We have shown that several classes of graphs, such as distance-hereditary graphs and $k$-leaf power graphs for low values of $k$ are \cL -graphs. We have also provided a complete characterization of the more restricted variant of monotone \cL -graphs by correspondence with the class of non-jumping graphs. This type of graph has a combinatorial description, expressed as the existence of a specific type of linear order of its vertices. 

The results of our paper suggest several open problems: 
What is the complexity of determining whether a given graph $G$ is a non-jumping? 
Are all planar graphs \cL -graphs?
Are $k$-leaf power graphs \cL -graphs for $k>4$? Our future work will investigate these questions.  

	\bibliographystyle{splncs03}

\newpage		\section *{Appendix}
	 
\textbf{Proof of Theorem~\ref{jumping_graph}}\\

\textbf{Theorem.} \textit{Not all graphs are non-jumping graphs.}\\

We prove this theorem by showing that the graph depicted in Fig.~\ref{fig:njumpingex}(d) is a non-jumping graph. We tested this graph with a computer program finding that all the possible $8!$ labeling are jumping. In this mathematical proof we use patterns that can occur in the labeling process to show that however a labeling is chosen, it is jumping.
\begin{wrapfigure}{R}{0.4\textwidth}
	\hfil \includegraphics[width=0.4\textwidth]{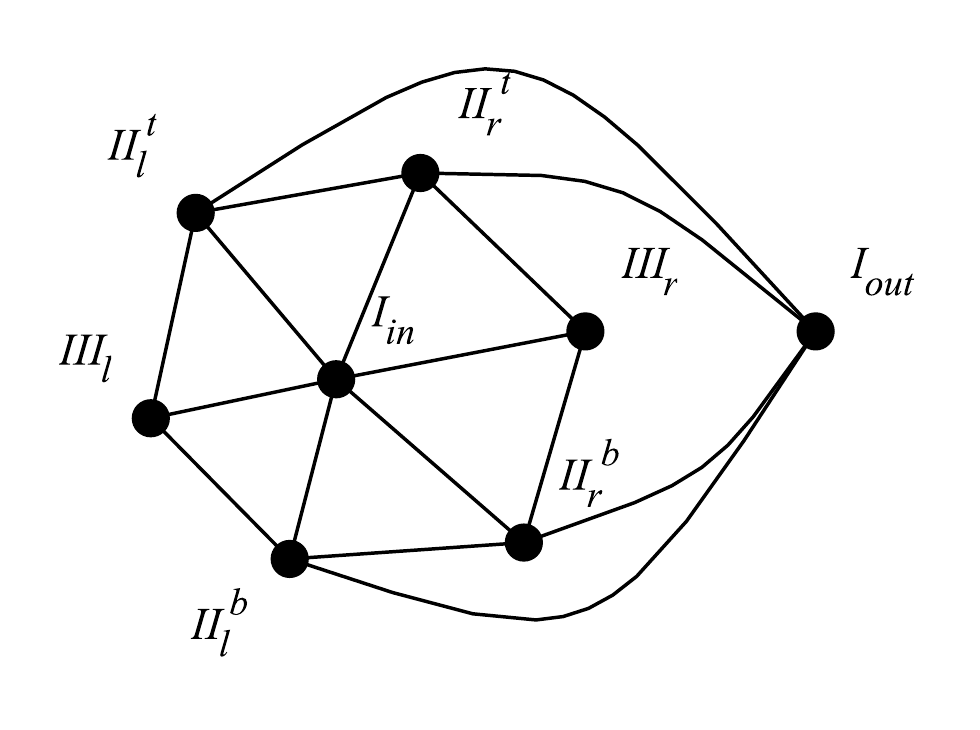} \hfil
	\caption{Notation for the vertices of the graph used in the proof of Theorem~\ref{jumping_graph}.}
	\label{fig:jumpingnotation}
\end{wrapfigure}

Before beginning the formal proof, we define the notation depicted in Fig.~\ref{fig:jumpingnotation}, to make easier the identification of jumping patterns.  Specifically, we provide each vertex with a name $X_{y}^{z}$ where $X \in \{I, II, III\}$ depending on its connection with other vertices: each $II$ is connected with two $I$'s, one $II$, and one $III$; each $III$ is connected with one $I$ and two $II$'s; and the $I$'s are the remaining vertices.  Next, we have $y \in \{l, r\}$ if $X \in \{II, III\}$  and $y \in \{in, out\}$ if $X = \{I\}$, such that $III_{l}$ ($III_{r}$) is connected with two $II_{l}$'s ($II_{r}$'s). We also have  $I_{in}$ connected with each $II$ and each $III$, while $I_{out}$ is connected with each $II$, and $I_{in}$ and $I_{out}$ are not connected. Finally, we have $z \in \{t, b\}$ indicating whether two $II$'s are connected \textemdash i.e., two $II$'s are adjacent iff they have the same value of $z$. 

To simplify our proof, in the following we use $h, k \in \{l, r\}$ where $h \ne k$ if not specified, and $i, j \in \{t, b\}$ where $i \ne j$ if not specified. Hence, for two vertices $II^i_h$ and $II^j_k$ we have:

\begin{itemize}
	\item $h=k \Leftrightarrow$ the vertices are adjacent to the same $III$.
	\item $i=j \Leftrightarrow$ the vertices are adjacent
\end{itemize}

To prove that every possible labeling is jumping, we first temporary remove the two vertices of type $I$. Their removal produces a cycle composed of six vertices.\\
Observe that a cycle has a non-jumping representation if we fix the label of one vertex and the other vertices are labeled in sequence, or if the labels of two adjacent vertices are swapped, i.e., the labels of any pair of adjacent vertices are at distance of at most two.
Thus, for our cycle the feasible (i.e., non-jumping) sequences are:
\begin{itemize}\label{li:cyclesq}
	\item $III_{h}$ $II_{h}^{i}$ $II_{k}^{i}$ $III_{k}$ $II_{k}^{j}$ $II_{h}^{j}$,
	\item $II_{h}^{j}$ $III_{h}$ $II_{h}^{i}$ $II_{k}^{i}$ $III_{k}$ $II_{k}^{j}$,
	\item $II_{k}^{j}$ $II_{h}^{j}$ $III_{h}$ $II_{h}^{i}$ $II_{k}^{i}$ $III_{k}$,
\end{itemize}
as well as all sequences obtained by permuting exactly two consecutive vertices. (For the sake of brevity, in the following we consider only the sequences without such a permutation. However, this proof can easily be extended to accommodate the permuted sequences.)

Before adding the two $I$'s to the sequences given above,
we consider the following infeasible, i.e. jumping, configuration. Note that we use the notation $(\dots)$ to signify to any vertex or sequence of vertices:
\begin{equation}
(\dots) II (\dots)  I (\dots) I  (\dots) II (\dots) 
\label{eq:twoismiddle}
\end{equation}
In general, the two $I$'s cannot be placed between any two pair of vertices of type $II$, since the $I$'s are not connected by an edge, but must be connected to every $II$. \\

From the infeasible configuration~\ref{eq:twoismiddle}, it follows that at least one $I$ should be placed to the left (or right) of all the $II$'s. Without loss of generality, we consider only placing $I$ to the left of the $II$'s. After placing one $I$ in this way, we have the following possibilities:
\begin{itemize}\label{li:=OneIsq}
	\item $I$ $III_{h}$ $II_{h}^{i}$ $II_{k}^{i}$ $III_{k}$ $II_{k}^{j}$ $II_{h}^{j}$
	\item $III_{h}$ $I$ $II_{h}^{i}$ $II_{k}^{i}$ $III_{k}$ $II_{k}^{j}$ $II_{h}^{j}$
	\item $I$ $II_{h}^{j}$ $III_{h}$ $II_{h}^{i}$ $II_{k}^{i}$ $III_{k}$ $II_{k}^{j}$
	\item $I$ $II_{k}^{j}$ $II_{h}^{j}$ $III_{h}$ $II_{h}^{i}$ $II_{k}^{i}$ $III_{k}$
\end{itemize}

We now observe that two nonadjacent $II$'s cannot be placed between the two $I$'s, that is,
\begin{equation}
(\dots) I (\dots)  II^{h} (\dots)  II^{k}  (\dots) I (\dots) 
\label{eq:twoinside}
\end{equation}
is unfeasible if $h \ne k$ is infeasible. This is because there is no an edge between $II^h$ and $II^k$, while there must be an edge between each $I$ and each $II$.  From this infeasible configuration, it follows that between the two $I$'s there can be only zero, one, or two $II$'s.

Because $III_k$ and $II_h$ do not share an edge, the sequence
\begin{equation}
(\dots) I (\dots)  III_k (\dots)   II_h (\dots)   II_k (\dots) 
\label{eq:iiikiih}
\end{equation}
is infeasible, and because $II^i_k$ and $II^j_k$ do not share an edge, we also have
\begin{equation}
(\dots) I (\dots) II_k (\dots) II_k (\dots) III_{k} (\dots)
\label{eq:iiikiihiiik}
\end{equation}
is infeasible. As such, we can eliminate the first and the last sequences in the previous list.  

The following list reports all the possible sequences with both $I$'s. Here, the available positions for the second $I$ that remain after considering the previous infeasible configurations are shown inside parentheses:
\begin{itemize}\label{li:=TwoIsq}
	\item $III_{h}$ $I$  ($I$)  $II_{h}^{i}$ ($I$) $II_{k}^{i}$ ($I$) $III_{k}$ ($I$) $II_{k}^{j}$ $II_{h}^{j}$
	\item $I$ ($I$) $II_{h}^{j}$ ($I$) $III_{h}$ ($I$) $II_{h}^{i}$ $II_{k}^{i}$ $III_{k}$ $II_{k}^{j}$
\end{itemize}

Now, if we identify the two $I$'s as $I_{in}$ and $I_{out}$, we find that the following are infeasible configurations:

\begin{equation}
(\dots) I_{in} (\dots) I_{out} (\dots) III (\dots) II (\dots)
\label{eq:inoutiiiii}
\end{equation}
because $I_{out}$ and $III$ do not share and edge, and

\begin{equation}
(\dots) II (\dots) III_{h} (\dots) I_{out} (\dots) II_{h} (\dots)
\label{eq:iiiiihioutii}
\end{equation}
because no $III$ and $I_{out}$ share an edge.  Thus, $II$ cannot be followed by a $III$, a $I_{out}$ and the adjacent $II$ of the previous $III$.

We observe that the previous configurations can occur in any of the two sequences. So the leftmost $I$ can only be $I_{out}$:
\begin{itemize}\label{li:=INOUT}
	\item $III_{h}$ $I_{out}$  ($I_{in}$)  $II_{h}^{i}$ ($I_{in}$) $II_{k}^{i}$ ($I_{in}$) $III_{k}$ ($I_{in}$) $II_{k}^{j}$ $II_{h}^{j}$
	\item $I_{out}$ ($I_{in}$) $II_{h}^{j}$ ($I_{in}$) $III_{h}$ ($I_{in}$) $II_{h}^{i}$ ($I_{in}$) $II_{k}^{i}$ $III_{k}$ $II_{k}^{j}$
\end{itemize}

We finally conclude that all the remaining sequences are jumping since they each contain at least one of the infeasible configurations:

\begin{equation}
(\dots) III (\dots) I_{out} (\dots) I_{in} (\dots) II (\dots)
\label{eq:iiioutinii}
\end{equation}
because $I_{out}$ and $I_{in}$ do not share an edge, and

\begin{equation}
(\dots) I (\dots) II_k^{i} (\dots) II_h^{i} (\dots) III_{k} (\dots)
\label{eq:iiikiihiiik}
\end{equation}
because $II_k^{i}$ and $II_h^{i}$ do not share an edge

It follows that the graph does not have a non-jumping labeling. As mentioned, this proof can easily be extended to any exchange of an adjacent pair of vertices of the cycle.\\
\\
{\bf  Proof of Theorem ~\ref{theorem:4-leaf-power}:}\\
     
\textbf{Theorem.} \textit{Every $4$-leaf power graph admits an \cL-embedding. }
 
\begin{figure}[!htbp]
	\hfil \vspace{-.6cm}\includegraphics[width=1\textwidth]{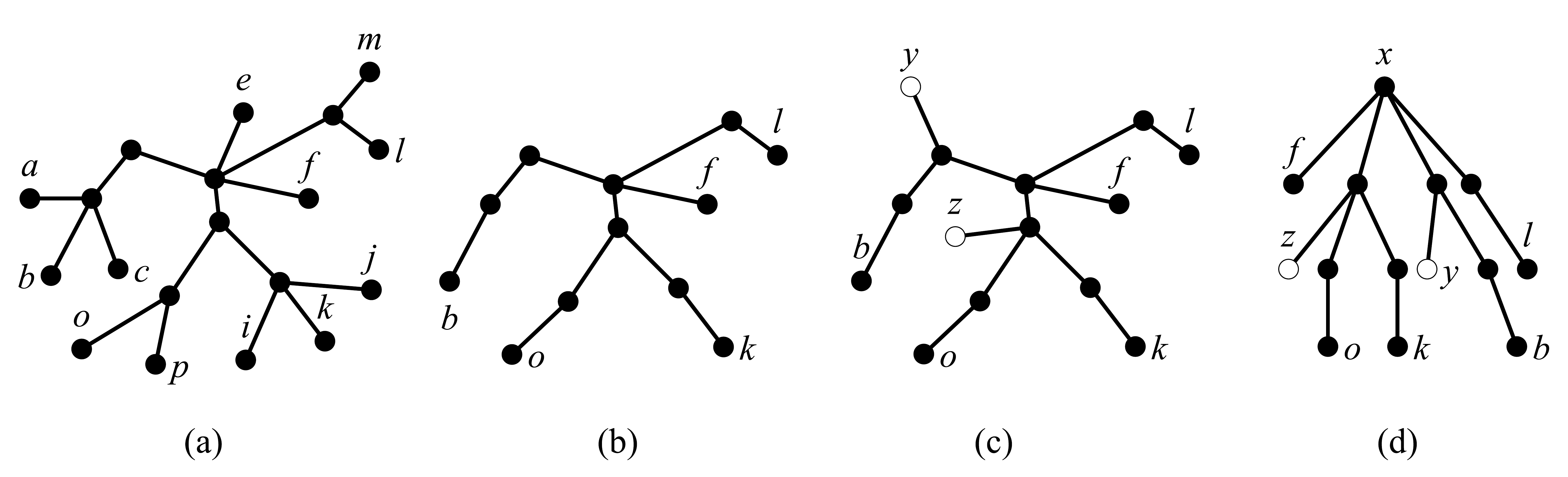}\vspace{-.2cm} \hfil
	\caption{(a) A tree $T$, which is modified by (b) removing multiple siblings, and (c) adding dummy leaves are to the internal vertices. (d) A rooted simplified leaf-tree of $T$.}
	\label{fig:4-leaf-tree-processing}
\end{figure}

Before we begin our proof, we first define some notation and terminology.  Let $G=(V,E)$ be a 4-leaf power graph. Then there is a tree $T'$ whose leaves correspond to the vertices of $G$ in such a way that two vertices are adjacent in $G$ precisely when their distance in $T'$ is at most 4. Let $v$ be a vertex of $G$. We denote the graph obtained by removing $v$ and all edges incident to $v$ by $G-v$. 

To make $T'$ as a simpler and more uniform tree, we remove and add some leaves on $T'$ as follows.  We first remove siblings leaves from $T'$, i.e., leaves that have the same parent,  and add new leaves to each internal node without any child-leaf to create a \textit{rooted simplified leaf-tree} $T$  of $G$ from $T'$, as follows. Let $u$ and $v$ be two siblings leaves in $T'$. Then $N(u)=N(v)$ in $G$ and $(u,v) \in E$. Hence, $u$ is a true twin of $v$ in $G$. According to Theorem~\ref{theorem:distance-hereditary}, if we have \cL($G-v$) then we have \cL($G$). For each group of sibling leaves, we keep exactly one leaf, removing the others. For example, Figure ~\ref{fig:4-leaf-tree-processing}(b) shows the transformed tree after removing multiple siblings from the tree shown in  Fig.~\ref{fig:4-leaf-tree-processing}(a). We now add a dummy leaf to the internal vertices of $T$ that do not have a child-leaf. In Figure~\ref{fig:4-leaf-tree-processing}(c), vertices $y$ and $z$ are dummy leaves. The dummy vertices will be removed from \cL($G$) after the construction of \cLE of  $G$. We make $T$ a rooted tree by selecting an arbitrary internal vertex as its root (see Fig.~\ref{fig:4-leaf-tree-processing}(d)). Observe that every internal vertex has exactly one leaf in the rooted simplified leaf-tree. 
     \begin{wrapfigure}[14]{R}{0.5\textwidth}
	\hfil \includegraphics[width=0.4\textwidth]{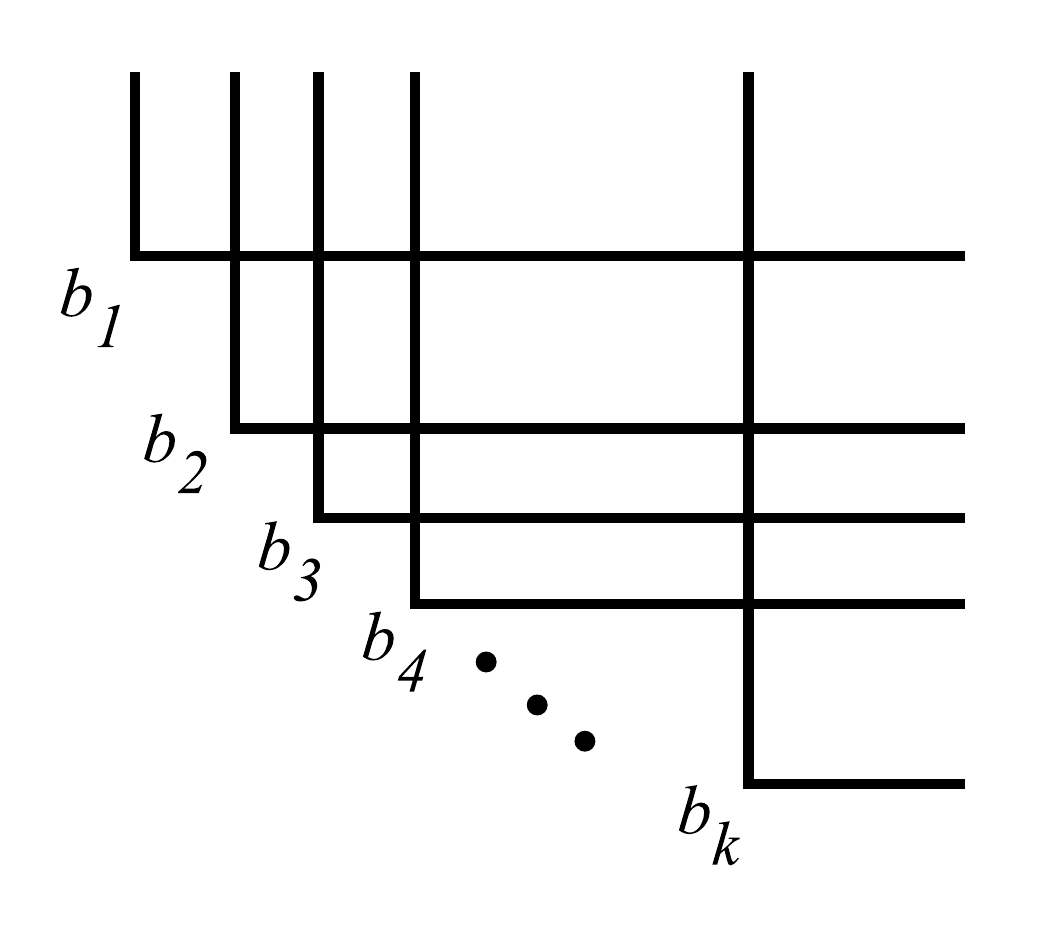} \hfil
	\caption{A fully connected \cLE}
	\label{fig:fullconnected}
\end{wrapfigure}
   
	Let $v$ be a vertex  of $T$. We denote the subtree rooted at $v$ by $T_v$. We denote the parent of $v$ by $v'$,  and the  parent of  $v'$ by  $v''$. Vertices $u$ and $v$ are said to be {\it siblings} if $u' = v'$. A vertex $u$ is said to be  an {\it uncle} of $v$ if $u'=v''$. A vertex $u$ is said to be a {\it p-uncle} of $v$ if $u'$ is parent of $v''$. Vertices $u$ and $v$ are said to be {\it cousins} if $u''=v''$. A vertex $u$ is said to be  a {\it nephew} of $v$ if $u''=v'$.
 The length of the largest path from the root of $T$ to any leaf is called {\it depth } of $T$. 
    
    Since $G$ is a 4-leaf power graph, then $(u,v) \in E$ if and only if one of the following is true:
    \begin{enumerate}
    \item $u$ and $v$ are siblings,
    \item $u$ and $v$ are cousins,
    \item $u$ is an uncle of $v$ or  $v$ is an uncle of $u$, or
    \item $u$ is a p-uncle of $v$ or $v$ is a p-uncle of $u$.
    \end{enumerate}
    By definition, a leaf $u$ has exactly one uncle and exactly one p-uncle. Also, cousins have a common uncle and a common p-uncle.


We now draw each group of cousins as a \textit{ fully connected \cLE} in which all of the \cL-segments intersect one another (see Fig.~\ref{fig:fullconnected}). 

  Let $r$ be the root of $T$.  Since $T$ is  a rooted simplified leaf-tree, $r$ has exactly one child-leaf, which we call $r_c$.  We denote a nephew of $r_c$ by $r_{cc}$, and the subtree induced by the vertices $r$, $r_c$, and $r_{cc}$ by $T_{conf}$. 

We maintain two configurations, a "\textit{Rectangle-configuration}" and an "\textit{L-configuration}," for the drawing of leaves of $T_{conf}$ of $T$.  The properties of a {\it Rectangle-configuration} are:
 \begin{description}
 \item [RconPro1] The horizontal arms of cousins intersect the vertical arm of their uncle, maintaining a subdivided L-shaped free region for each cousin
 \item [RconPro2] The vertical part of each subdivided L-shaped free region is visible from the horizontal arm of corresponding cousin
  \item [RconPro3] The horizontal part of every subdivided L-shaped free region is visible from the vertical arm of the uncle
 \end{description}
 
 The properties of an {\it L-configuration} are the following:
 \begin{description}
 \item [LconPro1] The vertical segments of cousins intersect the horizontal segment of their uncle, maintaining a subdivided rectangle shape-free region  for each cousin
 \item [LconPro2] The right part of each subdivided rectangle shape-free region is visible from the vertical arm of its corresponding cousin
  \item [LconPro3] The top part of every subdivided rectangle shape-free region is visible from the horizontal arm of the uncle
 \end{description}
 
Figure~\ref{fig:confs}(f) depicts a Rectangle-configuration for the tree drawn in  Fig.~\ref{fig:confs}(d), and Fig.~\ref{fig:confs}(e) depicts an L-configuration for the tree drawn in  Fig.~\ref{fig:confs}(d).

     We need the following lemma for the proof of Theorem~\ref{theorem:4-leaf-power} 
      
 \begin{lemma}
 \label{4-leaf-lemma}
    Let $G$ be a 4-leaf power graph and  $T$ be a corresponding rooted simplified leaf-tree of $G$.  Then $T_{conf}$ admits a Rectangle-configuration and an L-configuration.
 
 \end{lemma}

 \begin{proof}
 Let $r$ be the root of $T$. 
 Assume that $T_{conf}$ of $T$ contains $k+1$ leaves. These are $r_c$ and $k$  nephews of $r_c$ (for some $k\ge 1$); the $k$ nephews are cousins of each other. We find a Rectangle-configuration  of $T_{conf}$ as follows. We take a fully connected \cLE of the $k$ cousins and add \cL($r_c$) such that the horizontal arms of the $k$ cousins intersect the vertical arm of \cL($r_c$). It is easy to verify that the properties of the Rectangle-configuration hold (see Fig.~\ref{fig:confs}(b)).  
     Similarly, we find L-configuration  of $T$ as follows. We take a fully connected \cLE of $k$ cousins and add \cL($r_c$) such that the vertical arms of the $k$ cousins intersect the horizontal segment of \cL($r_c$). This drawing maintains the properties of L-configuration of $T$ (see Fig.~\ref{fig:confs}(c)). 
 \end{proof}

\begin{figure}[t]
	\hfil \vspace{-.4cm}\includegraphics[width=.95\textwidth]{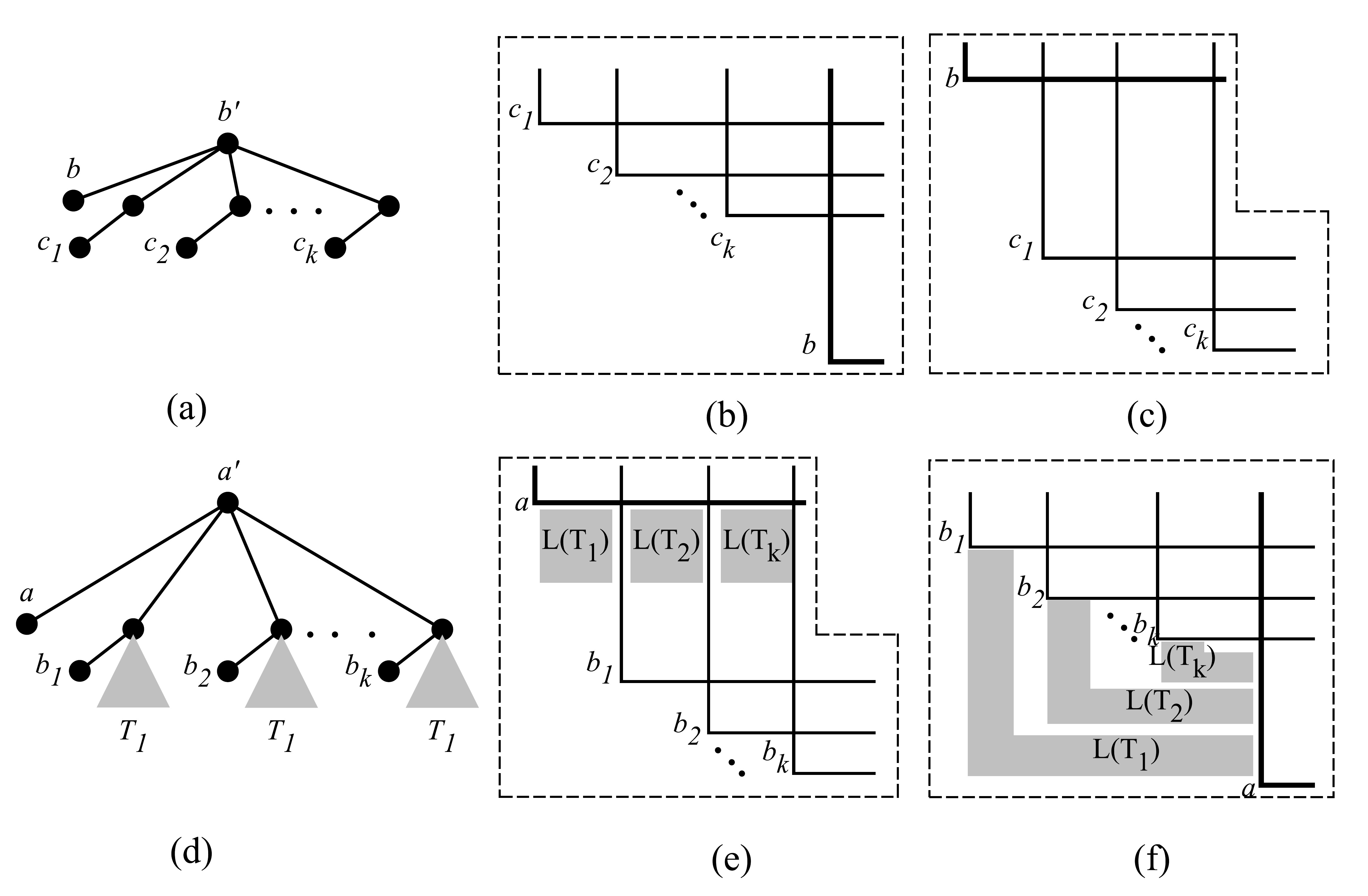}\vspace{-.2cm} \hfil
	\caption{(a) The tree $T$ used for the base case, (b) a Rectangle-configuration of $T$, and (c) an L-configuration of $T$. (d) The tree $T_{a'}$ used in the induction step, (e) an L-configuration of $T_{a'}$, and (f) a Rectangle-configuration of $T_{a'}$. }
	\label{fig:confs}
\end{figure}

 \begin{proof}  
    We now begin a proof of the main theorem by induction.
     Let $G_\rho$ be a 4-leaf power graph and  $T$ be a corresponding rooted simplified leaf-tree of $G_\rho$ with depth $\rho$. Let $r$ be the root of $T$.
We claim that $G_\rho$ admits two \cLE such that   $T_{conf}$   admits the Rectangle-configuration in one \cLE  and the L-configuration in the other \cLE.  Our induction is based on depth of $T$.  
     
     \textbf{Base case}: The depth $\rho=2$. Let $G_2$ be a 4-leaf power graph and  $T$  be a corresponding rooted simplified leaf-tree of $G_2$ with depth $2$.  Since $T_{conf}=T$, by Lemma~\ref{4-leaf-lemma} $G_2$ admits two \cL -embeddings such that   $T_{conf}$   admits the Rectangle-configuration in one \cLE  and the L-configuration in the other \cLE (see Fig.~\ref{fig:confs}(a)-(b)).
     
     \textbf{Induction case}: We now assume that our claim holds for any depth $\rho-1$. For the inductive step, we have to prove that it holds for depth $\rho$.
     
     Let $G_\rho$ be a 4-leaf power graph and  $T$  be a corresponding rooted simplified leaf-tree with depth $\rho$ of $G_\rho$. Let $r$ be the root of $T$ (Fig.~\ref{fig:confs}(a)). By Lemma~\ref{4-leaf-lemma}, the drawing of $T_{conf}$  maintains the properties of a Rectangle-configuration, but we need to show how to place \cL($T_{r_{cc}'}$) into  the L-shaped free region between \cL($r_{cc}$) and \cL($r_{c}$). (Any \cLE consists of horizontal and vertical line segments so width or height or both of the drawing can be resized.)
     
 Let $\rho'$ be the depth of $T_{r_{cc}'}$. Since $\rho'\le \rho-1$, by the inductive hypothesis, $G_{\rho'}$ has an L-configuration with respect to the root $r_{cc}$ of  $T_{r_{cc}'}$ with properties LconPro2 and LconPro3. Note that $r_{cc}$ of $T$ and ${r_{cc}'}_c$ of  $T_{r_{cc}'}$ are the same  vertices in $T$. We therefore place the corner of \cL(${r_{cc}'}_c$) on the corner of \cL$(r_{cc})$ to make one \cL-segment.  We know $r_{c}$ is the p-uncle of each ${r_{cc}'}_{cc}$. Hence \cL($r_{c}$) and \cL(${r_{cc}'}_{cc}$) should cross. The crossing between the nephews of $r_{cc}$ in  $T_{r_{cc}'}$ with $r_c$ in $T$ can be achieved by extending the horizontal arm of  \cL(${r_{cc}'}_{cc}$) because of LconPro2. Let $a$ and $b$ be two nephews of $r_c$. Then the distance between any leaf in $T_a$ and any leaf in $T_b$  is greater than 4, except for $a$ and $b$. Thus  $G_\rho$ admits an \cLE such that   $T_{conf}$   admits the Rectangle-configuration.

 We now show that $G_\rho$ admits an \cLE such that   $T_{conf}$ admits the \sloppy{L-configuration}.
      By Lemma~\ref{4-leaf-lemma}, the drawing of $T_{conf}$  maintains the properties of L-configuration  but we need to show how to  place \cL($T_{r_{cc}'}$) into  the L-shaped free region between \cL($r_{cc}$) and \cL($r_{c}$).  
     
 Let $\rho'$ be the depth of $T_{r_{cc}'}$. Since $\rho'\le \rho-1$,  by the inductive hypothesis we know that $G_{\rho''}$ has an Rectangle-configuration with respect to the root $r_{cc}$ of  $T_{r_{cc}'}$ with properties RconPro2 and RconPro3. Note that $r_{cc}$ of $T$ and ${r_{cc}'}_c$ of  $T_{r_{cc}'}$ are the same vertex. We place the corner of \cL(${r_{cc}'}_c$) on the corner of \cL$(r_{cc})$ to make one \cL-segment. We know $r_{c}$ is the p-uncle of each ${r_{cc}'}_{cc}$. Hence \cL($r_{c}$) and \cL(${r_{cc}'}_{cc}$) should cross. The crossing between the nephews of $r_{cc}$ in  $T_{r_{cc}'}$ to $r_c$ in $T$ can be done by extending vertical line segment of  \cL(${r_{cc}'}_{cc}$) because of LconPro2. Let $a$ and $b$ be two nephews of $r_c$. Then the distance between any leaf in $T_a$ and   any leaf in $T_b$ is greater than 4, except for $a$ and $b$. \qed
\end{proof}

    
	 \begin{figure}[t]
 \vspace{-.4cm}\includegraphics[width=.8\textwidth]{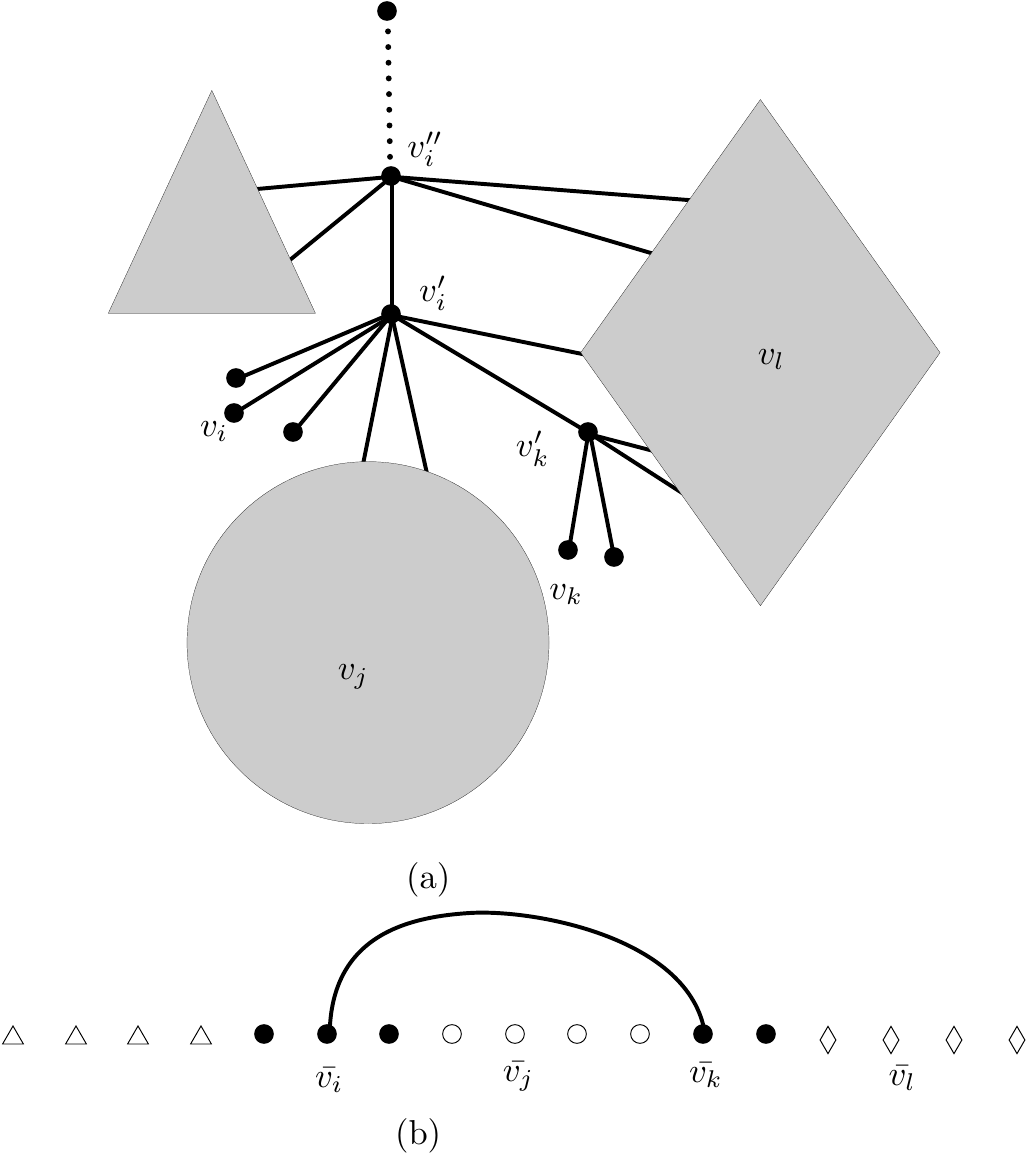} \hfil
	 	\caption{ Illustration for the proof of Theorem~\ref{non-jumping_3-leaf}.}
	 	\label{fig:3leaf-non-jumping-proof}
	 \end{figure}\vfill

\end{document}